\documentclass[conference]{IEEEtran}
\usepackage{amsmath}
\usepackage{amssymb}
\usepackage{stmaryrd}
\usepackage{epic,eepic}
\usepackage{theorem}
\usepackage{pifont}  %needed by dingautolist
\usepackage{euscript}
\usepackage{calc}
\usepackage[titles]{tocloft}
\usepackage[usenames]{color}          % Need the color package
\usepackage{xcolor}

\usepackage{epsf}           % dessin
\usepackage{graphicx}       % dessin
\usepackage{epsfig}         % dessin
\usepackage{pstricks}
\usepackage{pstricks-add}
\usepackage{pst-plot}
\usepackage{dsfont}
\usepackage{pst-circ}
\usepackage{pst-math}

\theoremheaderfont{\color{black}\normalfont\bfseries}

\newtheorem{lemma}{Lemma}
\newtheorem{theorem}{Theorem}
\newtheorem{definition}{Definition}

\newtheorem{proposition}{Proposition}
\newtheorem{remark}{Remark}
\newtheorem{conjecture}{Conjecture}

\newcommand{\argmin}{\ensuremath{\operatorname{argmin}}}

\newcommand{\R}{\mathbb{R}}

\newcommand{\N}{\mathbb{N}}
\newcommand{\E}{\mathbb{E}}
\newcommand{\Q}{\mathbb{Q}}

\newcommand{\PP}{\mathcal{P}}
\newcommand{\QQ}{\mathcal{Q}}
\newcommand{\mc}{\mathcal}

\newcommand{\prob}{\mathbb{P}}
\newcommand{\supp}{\mathrm{supp}\,}
\newcommand{\inte}{\ensuremath{\operatorname{int}}}

\ifCLASSINFOpdf
\else
\fi
\begin{document}
%
% paper title
% Titles are generally capitalized except for words such as a, an, and, as,
% at, but, by, for, in, nor, of, on, or, the, to and up, which are usually
% not capitalized unless they are the first or last word of the title.
% Linebreaks \\ can be used within to get better formatting as desired.
% Do not put math or special symbols in the title.
\title{Point-to-Point Strategic Communication}

% author names and affiliations
% use a multiple column layout for up to three different
% affiliations
\author{\IEEEauthorblockN{Ma\"{e}l Le Treust}
\IEEEauthorblockA{ETIS UMR 8051, CY Cergy Paris Universit\'{e}, ENSEA, CNRS,\\
6, avenue du Ponceau, \\
95014 Cergy-Pontoise CEDEX, France\\
Email: mael.le-treust@ensea.fr}
\and
\IEEEauthorblockN{Tristan Tomala}
\IEEEauthorblockA{HEC Paris, GREGHEC UMR 2959\\
1 rue de la Lib\'{e}ration,\\
78351 Jouy-en-Josas CEDEX, France\\
Email: tomala@hec.fr}}

% conference papers do not typically use \thanks and this command
% is locked out in conference mode. If really needed, such as for
% the acknowledgment of grants, issue a \IEEEoverridecommandlockouts
% after \documentclass

% for over three affiliations, or if they all won't fit within the width
% of the page, use this alternative format:
% 
%\author{\IEEEauthorblockN{Michael Shell\IEEEauthorrefmark{1},
%Homer Simpson\IEEEauthorrefmark{2},
%James Kirk\IEEEauthorrefmark{3}, 
%Montgomery Scott\IEEEauthorrefmark{3} and
%Eldon Tyrell\IEEEauthorrefmark{4}}
%\IEEEauthorblockA{\IEEEauthorrefmark{1}School of Electrical and Computer Engineering\\
%Georgia Institute of Technology,
%Atlanta, Georgia 30332--0250\\ Email: see http://www.michaelshell.org/contact.html}
%\IEEEauthorblockA{\IEEEauthorrefmark{2}Twentieth Century Fox, Springfield, USA\\
%Email: homer@thesimpsons.com}
%\IEEEauthorblockA{\IEEEauthorrefmark{3}Starfleet Academy, San Francisco, California 96678-2391\\
%Telephone: (800) 555--1212, Fax: (888) 555--1212}
%\IEEEauthorblockA{\IEEEauthorrefmark{4}Tyrell Inc., 123 Replicant Street, Los Angeles, California 90210--4321}}

% use for special paper notices
%\IEEEspecialpapernotice{(Invited Paper)}

% make the title area
\maketitle

%Point-to-Point Strategic Communication
\begin{abstract}
We propose a strategic formulation for the joint source-channel coding problem in which the encoder and the decoder are endowed with distinct distortion functions. We provide the solutions in four different scenarios. First, we assume that the encoder and the decoder cooperate in order to achieve a certain pair of distortion values. Second, we suppose that the encoder commits to a strategy whereas the decoder implements a best response, as in the persuasion game where the encoder is the Stackelberg leader. Third, we consider that the decoder commits to a strategy, as in the mismatched rate-distortion problem or as in the mechanism design framework. Fourth, we investigate the cheap talk game in which the encoding and the decoding strategies form a Nash equilibrium.
\end{abstract}

\IEEEpeerreviewmaketitle

\section{Introduction}

Strategic communication takes place when an informed sender communicates with  a receiver that takes an action, given that the sender and the receiver optimize different metrics. This question was originally formulated in the game theory literature were the messages are costless and the communication is unrestricted. Crawford and Sobel \cite{CrawfordSobel1982StrategicInformation} investigate the Nash equilibrium of the cheap talk game, whereas Kamenica and Gentzkow \cite{KamenicaGentzkow11} introduce the Bayesian persuasion game in which the sender commits to an information disclosure policy, as the leader of the Stackelberg game. In a previous work \cite{LeTreustTomala19}, we characterize the solution of the Bayesian persuasion game when the communication channel is noisy.

\begin{figure}[!ht]
\begin{center}
\psset{xunit=0.9cm,yunit=0.9cm}
\begin{pspicture}(0,0.3)(8.5,1.3)
\pscircle(0,0.5){0.45}
\psframe(2,0)(3,1)
\pscircle(5,0.5){0.45}
\psframe(7,0)(8,1)
\psline[linewidth=1pt]{->}(0.5,0.5)(2,0.5)
\psline[linewidth=1pt]{->}(3,0.5)(4.5,0.5)
\psline[linewidth=1pt]{->}(5.5,0.5)(7,0.5)
\psline[linewidth=1pt]{->}(8,0.5)(9,0.5)
\rput[u](1,0.8){$U^{n}$}
\rput[u](3.75,0.8){$X^n$}
\rput[u](6.25,0.8){$Y^n$}
\rput[u](8.5,0.8){$V^n$}
\rput(0,0.5){$\PP_U$}
\rput(5,0.5){$\mc{T}_{Y|X}$}
\rput(2.5,0.5){$\sigma$}
\rput(7.5,0.5){$\tau$}
\rput(2.5,1.3){$d_{\textsf{e}}(u,v)$}
\rput(7.5,1.3){$d_{\textsf{d}}(u,v)$}
\end{pspicture}
\caption{The source is i.i.d. and the channel is memoryless. The encoder and the decoder have mismatched distortion functions $d_{\textsf{e}}(u,v) \neq d_{\textsf{d}}(u,v)$.} %The encoder and the decoder implement strategies $\sigma$ and $\tau$ in order to minimize mismatched distortion functions $d_{\textsf{e}}(u,v) \neq d_{\textsf{d}}(u,v)$. }
\label{fig:SystemModel}
\end{center}
\end{figure}
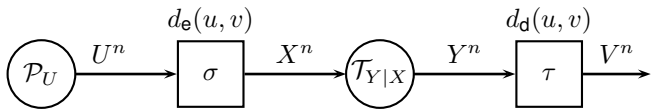

\vspace{-0.4cm}

The strategic communication problem has attracted attention in computer science \cite{dughmi2019persuasion}, in control theory \cite{SaritasYukselGezici2020}, in information theory \cite{AkyolLangbortBasarIEEE17}, \cite{LeTreustTomala(Allerton)16}, \cite{LeTreustTomala(IZS)18}, \cite{LeTreustTomala_IT19} and is related to the lossy source coding with mismatch distortion functions \cite{Lapidoth97}, \cite{ScarlettGuillenSomekhMartinez_FTCIT20}. Recently, Vora and Kulkarni investigate a strategic communication problem in which the receiver is the Stackelberg leader that should recover the source sequence \cite{VoraKulkarni_ISIT20}. The authors introduce the notion of the ``information extraction capacity'' and formulate an elegant solution in terms of the zero error capacity of ``the sender graph'' \cite{VoraKulkarni_Arxiv2020}.
 
In this paper, we compare four different solutions for the point-to-point strategic communication problem, and we characterize the set of Nash equilibrium distortions. 
%The main result of this paper, Theorem \ref{theo:DecoderCommitment}, characterizes the solution of a mechanism design problem, similar to the setting of Jackson and Sonnenschein \cite{JacksonSonnenschein07}, in which the channel noise  corrupts the transmitted signals.

%%%%%%%%%%%%%%%%%%%%%%%%%%%%%%%%%%%%%%%%%%
%%%%%%%%%%%%%%%%%%%%%%%%%%%%%%%%%%%%%%%%%%
%\newpage
\section{System model}

We denote by $\mc{U}$, $\mc{X}$, $\mc{Y}$, $\mc{V}$, the finite sets of information source, channel inputs, channel outputs and decoder's outputs. Uppercase letters $U^n=(U_1,\ldots,U_n)\in\mc{U}^n$ and $X^n$, $Y^n$, $V^n$ stand for $n$-length sequences of random variables with $n\in \mathbb{N}^{\star}=\mathbb{N}\setminus\{0\}$, whereas lowercase letters $u^n=(u_1,\ldots,u_n)\in\mc{U}^n$ and $x^n$, $y^n$, $v^n$, stand for sequences of realizations.  We denote by $\Delta(\mc{X})$ the set of  probability distributions $\QQ_X$ over $\mc{X}$, i.e. the probability simplex. 
%For a distribution $\QQ_{X}\in  \Delta(\mc{X})$, we write $\QQ(x)$ instead of $\QQ_X(x)$ for the probability value assigned to realization $x\in \mc{X}$. The notation $\QQ_{X}(\cdot|y)\in \Delta(\mc{X})$ denotes the conditional distribution of $X\in\mc{X}$ given the realization $y\in \mc{Y}$ and $\QQ_{X}^{\otimes n}\in \Delta(\mc{X}^n)$ denotes the i.i.d. distribution. 
% The support of a distribution $\QQ_{X}$ is denoted by $\supp \QQ_{X}=\{x\in\mc{X},\;\QQ(x)>0\}$ and $\inte \Delta(\mc{X})$ denotes the relative interior of the probability simplex.  The distance between two distributions $\QQ_X$ and $\PP_X$ is based on $L^1$ norm, denoted by $||\QQ_X - \PP_X||_{1}= \sum_{x\in\mc{X}} |\QQ(x) - \PP(x)|$. We denote by $D(\QQ_X||\PP_X)$ the Kullback-Leibler (K-L) divergence. The notation $U  -\!\!\!\!\minuso\!\!\!\!-X    -\!\!\!\!\minuso\!\!\!\!-  Y$ stands for the Markov chain property corresponding to $\PP_{Y|XU} = \PP_{Y|X}$. 
We consider an i.i.d. information source and a memoryless channel distributed according to $\PP_{U}\in \Delta(\mc{U})$ and $\mc{T}_{Y|X} : \mc{X} \to \Delta(\mc{Y})$, as depicted in Fig. \ref{fig:SystemModel}.

%%%%%%%%%%%%%%%%%%%
%%%%expressions
%%%%%%%%%%%%%%%%%%%
%probability mass function PMF
%
%probability kernel
%transition probability	-> 	Markov property
%conditional probablity
%stochastic mapping
%transition kernel

\begin{definition}\label{def:Code}
We define the encoding strategy $\sigma:\mc{U}^{n} \longrightarrow \Delta(\mc{X}^n)$ and the decoding strategy $\tau : \mc{Y}^n   \longrightarrow  \Delta( \mc{V}^n)$, and we denote by $\PP^{\sigma,\tau}$ the distribution defined by
\begin{align}
\PP^{\sigma,\tau}= \bigg(\prod_{t=1}^n\PP_{U_t} \bigg)\sigma_{X^n|U^n}
   \bigg(\prod_{t=1}^n \mc{T}_{Y_t|X_t} \bigg) \tau_{V^n|Y^n},
\end{align}
where $\sigma_{X^n|U^n}$, $\tau_{V^n|Y^n}$ denote the distributions of $\sigma$, $\tau$.
\end{definition}

%We assume that the encoder and the decoder minimize distincts distortion functions. 
\begin{definition}\label{def:Distortions} 
The encoder and decoder distortion functions $d_{\textsf{e}} : \mc{U} \times \mc{V} \longrightarrow \R$ and $d_{\textsf{d}} : \mc{U} \times \mc{V} \longrightarrow \R$ induce long-run distortion functions $d_{\textsf{e}}^{\,n}(\sigma,\tau)$ and $d_{\textsf{d}}^{\,n}(\sigma,\tau)$  defined by
%We denote by $d_{\textsf{e}}^{\,n}(\sigma,\tau)$ and $d_{\textsf{d}}^{\,n}(\sigma,\tau)$ the long-run distortion functions, evaluated with respect to $\PP^{\sigma,\tau}$ and defined by
\begin{align}
%d_{\textsf{e}}^{\,n}(\sigma,\tau) =& \E_{\sigma,\tau} \Bigg[ \frac{1}{n} \sum_{t=1}^n d_{\textsf{e}}(U_t,V_t) \Bigg] ,\\
d_{\textsf{d}}^{\,n}(\sigma,\tau) =& \sum_{u^n,v^n}\PP^{\sigma,\tau}\big(u^n,v^n \big) \cdot  \Bigg[  \frac{1}{n} \sum_{t=1}^n d_{\textsf{d}}(u_t,v_t) \Bigg].
\end{align}
\end{definition}

%%%%%%%%%%%%%%%%%%%%%%%%%%%%%%%%%%%%%%%%%%
%%%%%%%%%%%%%%%%%%%%%%%%%%%%%%%%%%%%%%%%%%
\section{Cooperative scenario}\label{sec:Cooperative}

\begin{definition}
The pair $(D_{\textsf{e}},D_{\textsf{d}})$ is achievable if 
\begin{align}
&\forall \varepsilon>0,  \;\; \exists \bar{n}\in \N^{\star}, \;\;\forall n\geq \bar{n}, \;\;\exists (\sigma,\tau) \\
&\text{ s.t. }\quad |D_{\textsf{e}} - d_{\textsf{e}}^{\,n}(\sigma, \tau)|+ |D_{\textsf{d}} - d_{\textsf{d}}^{\,n}(\sigma, \tau)|\leq \varepsilon
\end{align}
We denote by $\mc{C} $ the set of 	achievable pairs $(D_{\textsf{e}},D_{\textsf{d}})$.
\end{definition}
We define the set of distributions
\begin{align}
\Q_1 = \Big\{\PP_U\QQ_{V|U} \text{ s.t. } \max_{\PP_X} I( X; Y )  -   I( U ;V)   \geq 0\Big\}.
\end{align}
\begin{theorem}[Cooperative scenario]\label{theo:CooperativeScenario}
\begin{align}
\mc{C} =& \Big\{\big(\E_{\QQ}[d_{\textsf{e}}(U,V)], \E_{\QQ}[d_{\textsf{d}}(U,V)]\big) \quad \QQ\in\Q_1 \Big\}.
\end{align}
\end{theorem}
The proof of Theorem \ref{theo:CooperativeScenario} follows from Shannon's separation result \cite[Theorem 3.7]{ElGammalKim(book)11}, with two distortion functions.

%%%%%%%%%%%%%%%%%%%%%%%%%%%%%%%%%%%%%%%%%%
%%%%%%%%%%%%%%%%%%%%%%%%%%%%%%%%%%%%%%%%%%
%\newpage
\section{Persuasion game: encoder commitment}
In this section, the encoder chooses first a strategy $\sigma$, and the decoder selects a best response strategy $\tau$ accordingly. This corresponds to the Bayesian persuasion game \cite{KamenicaGentzkow11}, where the encoder is the Stackelberg leader.
\begin{definition}\label{def:DecBR} Given $n\in\N^{\star}$, we define\\
1. the set of decoder best responses to strategy $\sigma$ by
\begin{align}
 \textsf{BR}_{\textsf{d}}(\sigma) =&\underset{\tau}{\argmin}\; d_{\textsf{d}}^{\,n}(\sigma, \tau),\end{align}
2. the long-run encoder distortion value by
\begin{align}
D^n_{\textsf{e}}=\inf_{\sigma}\max_{\tau \in \textsf{BR}_{\textsf{d}}(\sigma)} d_{\textsf{e}}^{\,n}(\sigma, \tau). \label{eq:PersuasionProblem}
\end{align}
\end{definition}

In case $\textsf{BR}_{\textsf{d}}(\sigma)$ is not a singleton, we assume that the decoder selects the worst strategy for the encoder distortion  $\max_{\tau \in \textsf{BR}_{\textsf{d}}(\sigma)} d_{\textsf{e}}^{\,n}(\sigma, \tau)$, so that the solution is robust to the exact specification of the decoding strategy. 

We aim at characterizing the asymptotic behavior of $D^n_{\textsf{e}}$.
\begin{definition}\label{def:CharacterizationPersuasion} 
We consider an auxiliary random variable $W\in \mc{W}$ with $|\mc{W}| = \min\big(|\mc{U}|+1, |\mc{V}|\big)$ and we define
\begin{align}
\Q_2 =& \Big\{  \PP_{U}  \QQ_{W|U}   \; \text{s.t.} \; \max_{\PP_X} I( X; Y )  -   I( U ;W)   \geq 0  \Big\}.\label{eq:SetQ0}
\end{align} 
Given $\QQ_{UW}$, we define the single-letter decoder best responses
\begin{align}
\Q_{\textsf{d}}\big(\QQ_{UW}\big) =& \underset{\QQ_{V|W}}{\argmin}\; \E_{\QQ_{UW}\atop  \QQ_{V|W}  } \Big[ d_{\textsf{d}}(U,V) \Big].\label{eq:SetQ2}
\end{align} 
The encoder optimal distortion $D_{\textsf{e}}^{\star} $ is given by 
\begin{align}
D_{\textsf{e}}^{\star} =&  \inf_{\QQ_{UW} \in \Q_2} \max_{\QQ_{V|W}  \in  \atop \Q_{\textsf{d}}(\QQ_{UW})} \E_{\QQ_{UW} \atop  \QQ_{V|W} } \Big[d_{\textsf{e}}(U,V)\Big].\label{eq:SolutionEncoder}
\end{align}
\end{definition}

 \begin{theorem}[Encoder commitment, Theorem 3.1 in \cite{LeTreustTomala19}]\label{theo:EncoderCommitment}
%The encoder long-run distortion satisfies:
\begin{align}
\forall n \in \N^{\star}, \qquad &D^n_{\textsf{e}}  \geq  D_{\textsf{e}}^{\star},\label{eq:Converse}\\
\forall \varepsilon>0, \;  \exists \bar{n}\in \N^{\star},\; \forall n\geq \bar{n}, \qquad &D^n_{\textsf{e}} \leq  D_{\textsf{e}}^{\star}  +  \varepsilon.\label{eq:Achievability}
\end{align}
\end{theorem}

%Theorem \ref{theo:EncoderCommitment} is equivalent to \cite[Theorem 3.1]{LeTreustTomala19}, which consider the maximization of utility functions instead of the minimization of distortion functions. 
%We provide a sketch of proof in App. \ref{sec:ProofTheoEncoder}. 
Theorem \ref{theo:EncoderCommitment} is a particular case of \cite[Theorem III.3]{LeTreustTomala_IT19} when no side information is available at the decoder. Note that the sequence $(D^n_{\textsf{e}})_{n\in\N^{\star}}$ is sub-additive. Indeed, when $\sigma$ is the concatenation of several encoding strategies, the concatenation of the corresponding optimal decoding strategies still belongs to $ \textsf{BR}_{\textsf{d}}(\sigma)$. Theorem \ref{theo:EncoderCommitment} and Fekete's lemma, show that
\begin{align}
D_{\textsf{e}}^{\star} =& \lim_{n \to +\infty} D_{\textsf{e}}^n =\inf_{n  \in \N^{\star}} D_{\textsf{e}}^n .\label{eq:EncoderLimit}
\end{align}

\begin{remark}
The decoder long-run distortion $d_{\textsf{d}}^{\,n}(\sigma, \tau)$ obtained with $\sigma$ asymptotically optimal for \eqref{eq:PersuasionProblem} and $\tau \in \textsf{BR}_{\textsf{d}}(\sigma)$ converges to $\E_{\QQ_{UW} \atop  \QQ_{V|W} } \big[d_{\textsf{d}}(U,V)\big]$, where $\QQ_{V|W}\in \Q_{\textsf{d}}\big(\QQ_{UW}\big) $ and  $\QQ_{UW}$ is a limit of a minimizing sequence of \eqref{eq:SolutionEncoder}.
\end{remark}

%%%%%%%%%%%%%%%%%%%%%%%%%%%%%%%%%%%%%%%%%%
%%%%%%%%%%%%%%%%%%%%%%%%%%%%%%%%%%%%%%%%%%
\section{Mechanism design: decoder commitment}\label{sec:MD}
In this section, it is the decoder which chooses first a strategy $\tau$, and then the encoder selects a strategy $\sigma$ accordingly. This corresponds to the mismatched rate-distortion problem in information theory \cite{Lapidoth97}, \cite{ScarlettGuillenSomekhMartinez_FTCIT20}, and to the Mechanism design problem \cite{JacksonSonnenschein07} in game theory, where the decoder is the Stackelberg leader.

\begin{definition}\label{def:EncBR} Given $n\in\N^{\star}$, we define\\
1.  the set of encoder best responses to strategy $\tau$ by
\begin{align}
\textsf{BR}_{\textsf{e}}(\tau) =&\underset{\sigma}{\argmin}\; d_{\textsf{e}}^{\,n}(\sigma, \tau),
\end{align}
2. the long-run decoder distortion value by
\begin{align}
D^n_{\textsf{d}}=\inf_{\tau}\max_{\sigma  \in \textsf{BR}_{\textsf{e}}(\tau)} d_{\textsf{d}}^{\,n}(\sigma, \tau). \label{eq:MDProblem}
\end{align}
\end{definition}

The value $D^n_{\textsf{d}}$ corresponds to the best distortion the decoder can obtain for fixed $n\in\N^{\star}$. In case there are several  best responses, we assume the encoder selects the worst strategy $\sigma$ for the decoder distortion.

We aim at characterizing the asymptotic behaviour of $D^n_{\textsf{d}}$

%In case $\textsf{BR}_{\textsf{e}}(\tau)$ is not a singleton, we assume that the encoder selects the worst strategy for the decoder distortion  $\max_{\sigma  \in \textsf{BR}_{\textsf{e}}(\tau)} d_{\textsf{d}}^{\,n}(\sigma, \tau)$, so that the solution is robust to the exact specification of encoding strategy. 

%Note that the above defined sequence is sub-additive. Indeed, when $\tau$ is the concatenation of several decoding strategies, the concatenation of the corresponding optimal encoding strategies still belongs to $\textsf{BR}_{\textsf{e}}(\tau)$.

\begin{definition}\label{def:CharacterizationMD} 
Given an auxiliary random variable $W\in \mc{W}$ with $|\mc{W}| = \min\big(|\mc{U}|+1, |\mc{V}|\big)$ with distribution $\PP_{W}$, we define
\begin{align}
&\Q_3(\PP_{W}) =  \Big\{ \QQ_{UW}\in\Delta(\mc{U}\times \mc{W}) \; \text{s.t.} \;\QQ_{U} = \PP_{U},\nonumber\\
&\;\;\; \QQ_{W} = \PP_{W}\;\; \text{and} \;\; \max_{\PP_X} I( X; Y )  -   I( U ;W)   \geq 0\Big\}.
\end{align}
Given $\PP_{WV}$, we define the single-letter encoder best responses
\begin{align}
&\Q_{\textsf{e}}(\PP_{WV}) = \underset{\QQ_{UW}\in\Q_3(\PP_{W})}{\argmin}\;   \E_{\QQ_{UW} \atop  \PP_{V|W} } \Big[d_{\textsf{e}}(U,V)\Big].
\end{align}
The decoder optimal distortion $D_{\textsf{d}}^{\star} $ is given by 
\begin{align}
D_{\textsf{d}}^{\star} =&  \inf_{\PP_{WV}} \max_{\QQ_{UW}\in\Q_{\textsf{e}}(\PP_{WV})} \E_{\QQ_{UW} \atop  \PP_{V|W} } \Big[d_{\textsf{d}}(U,V)\Big].\label{eq:SolutionDecoder}
\end{align}
\end{definition}

In both \eqref{eq:SolutionEncoder} and \eqref{eq:SolutionDecoder}, it is the Stackelberg leader that selects the marginal distribution $\PP_{W}$,  whereas the incentive constraints affect the Stackelberg follower.  Furthermore, the encoder selects the distribution $\QQ_{UW} \in \Q_3(\PP_{W})$ that satisfies the information constraint and the decoder selects $\PP_{V|W}$.  
%Then, we could also write
%\begin{align}
%D_{\textsf{e}}^{\star} =&  \inf_{P_W}\inf_{\QQ_{UW} \in \Q_3(\PP_{W}) } \max_{\QQ_{V|W}\atop + \textsf{incentives}(\QQ_{UW})}  \E_{\QQ_{UW} \atop  \QQ_{V|W} } \Big[d_{\textsf{e}}(U,V)\Big],\label{eq:SolutionEncoder2}\\
%D_{\textsf{d}}^{\star} =&  \inf_{P_W} \inf_{\PP_{V|W}} \max_{\QQ_{UW}\in \Q_3(\PP_{W}) \atop + \textsf{incentives}(\PP_{WV})} \E_{\QQ_{UW} \atop  \PP_{V|W} } \Big[d_{\textsf{d}}(U,V)\Big].\label{eq:SolutionDecoder2}
%\end{align}

\begin{theorem}[Decoder commitment]\label{theo:DecoderCommitment}
\begin{align}
\forall n \in \N^{\star}, \qquad & D^n_{\textsf{d}} \geq    D_{\textsf{d}}^{\star},\label{eq:ConverseMD}\\
\forall \varepsilon>0,  \; \exists \bar{n}\in \N^{\star},\;\forall n\geq \bar{n}, \qquad &D^n_{\textsf{d}} \leq  D_{\textsf{d}}^{\star}  +  \varepsilon.\label{eq:AchievabilityMD}
\end{align}
\end{theorem}

The achievability proof of Theorem \ref{theo:DecoderCommitment} is provided in App. \ref{sec:ProofTheoDecoder}, and relies on similar arguments as in \cite[Step 1]{Lapidoth97} and \cite[Lemma 4.3]{ScarlettGuillenSomekhMartinez_FTCIT20}. The converse proof is based on standard arguments with the identification of the auxiliary random variable $W=(Y^{T-1},Y_{T+1}^n,T)$,  $T\in\{1,\ldots,n\}$. The sequence $(D^n_{\textsf{d}})_{n\in\N^{\star}}$ is sub-additive, thus Theorem \ref{theo:DecoderCommitment} and Fekete's lemma show that 
%\eqref{eq:MDProblem} converges to its infimum.
\begin{align}
D_{\textsf{d}}^{\star} =& \lim_{n \to +\infty} D^n_{\textsf{d}} = \inf_{n  \in \N^{\star}} D^n_{\textsf{d}}. \label{eq:DecoderLimit}
\end{align}

\begin{remark}
The encoder long-run distortion $d_{\textsf{e}}^{\,n}(\sigma, \tau)$ obtained with $\tau$ asymptotically optimal for \eqref{eq:MDProblem} and $\sigma \in \textsf{BR}_{\textsf{e}}(\tau)$ converges to $\E_{\QQ_{UW} \atop  \PP_{V|W} } \big[d_{\textsf{e}}(U,V)\big]$, where $\QQ_{UW}\in \Q_{\textsf{e}}\big(\PP_{WV}\big) $ and  $\PP_{WV}$ is a limit of a minimizing sequence of \eqref{eq:SolutionDecoder}.
\end{remark}

%%%%%%%%%%%%%%%%%%%%%%%%%%%%%%%%%%%%%%%%%%
%%%%%%%%%%%%%%%%%%%%%%%%%%%%%%%%%%%%%%%%%%
\section{Cheap talk game: no commitment}

%No-commitment
%Nash Equilibrium

\begin{definition}
Given $\varepsilon\geq0$ and $n\in \N^{\star}$, an $\varepsilon$-Nash equilibrium is a pair of strategies  $(\sigma,\tau)$ such that 
\begin{align}
&\sigma \in \textsf{BR}_{\textsf{e}}^{\,\varepsilon}(\tau)\quad \text{ and }\quad \tau \in \textsf{BR}^{\,\varepsilon}_{\textsf{d}}(\sigma)\quad \text{ where,}\\
%\end{align}
%where 
%\begin{align}
&\textsf{BR}_{\,\textsf{e}}^{\varepsilon}(\tau) =\Big\{\sigma,\;\; d_{\textsf{e}}^{\,n}(\sigma, \tau) \leq \min_{\tilde{\sigma}}d_{\textsf{e}}^{\,n}(\tilde{\sigma}, \tau) + \varepsilon\Big\},\\
&\textsf{BR}_{\,\textsf{d}}^{\varepsilon}(\sigma) =\Big\{\tau, \;\;d_{\textsf{d}}^{\,n}(\sigma, \tau) \leq \min_{\tilde{\tau}}d_{\textsf{d}}^{\,n}(\sigma, \tilde{\tau}) + \varepsilon\Big\}.
\end{align}
We denote by $\textsf{NE}_{\varepsilon}^{\,n}$ the set of distortion pairs $(D_{\textsf{e}}^{\varepsilon},D_{\textsf{d}}^{\varepsilon})$ for which there exists a $\varepsilon$-Nash equilibrium $(\sigma,\tau)$ such that 
\begin{align}
D_{\textsf{e}}^{\varepsilon} = d_{\textsf{e}}^{\,n}(\sigma, \tau)\quad \text{ and }\quad D_{\textsf{d}}^{\varepsilon} = d_{\textsf{d}}^{\,n}(\sigma, \tau).
\end{align}
We denote by $\textsf{NE}^{\,n}$ the set of $\textsf{NE}_{\varepsilon}^{\,n}$ with $\varepsilon=0$. 
\end{definition}

\begin{definition}
For $\varepsilon\geq0$, we define the set of distributions that are $\varepsilon$-best responses for both encoder and decoder.
\begin{align}
&\Q_4^{\varepsilon} = \Big\{ \QQ_{UWV} = \PP_{U}  \QQ_{W|U}  \QQ_{V|W} \quad \text{ s.t. }\nonumber\\
 &\qquad\quad \QQ_{UW}\in\Q_{\textsf{e}}^{\varepsilon}(\QQ_{WV}),\quad  \QQ_{V|W}  \in  \Q_{\textsf{d}}^{\varepsilon}(\QQ_{UW})\Big\},\\
%\end{align}
%where 
%\begin{align}
&\Q_{\textsf{e}}^{\varepsilon}(\QQ_{WV}) = \Big\{ \QQ_{UW}\in\Q_3(\QQ_{W})\;\text{ s.t. } \E_{\QQ_{UW} \atop  \QQ_{V|W} } \Big[d_{\textsf{e}}(U,V)\Big] \nonumber \\
&\qquad\qquad\quad\leq \min_{ \widetilde{\QQ}_{UW}\atop\in\Q_3(\QQ_{W})}\E_{\widetilde{\QQ}_{UW} \atop  \QQ_{V|W} } \Big[d_{\textsf{e}}(U,V)\Big]+\varepsilon \Big\},\\
&\Q_{\textsf{d}}^{\varepsilon}(\QQ_{UW})= \Big\{\QQ_{V|W} \text{ s.t. } \E_{\QQ_{UW} \atop  \QQ_{V|W} } \Big[d_{\textsf{d}}(U,V)\Big] \nonumber \\
&\qquad\qquad\quad\leq \min_{ \widetilde{\PP}_{V|W}} \E_{\QQ_{UW} \atop  \widetilde{\PP}_{V|W} } \Big[d_{\textsf{d}}(U,V)\Big]+\varepsilon \Big\}.
\end{align}
Then, we define
\begin{align}
\mc{N}^{\varepsilon} = \Big\{&\big(\E_{\QQ}[d_{\textsf{e}}(U,V)], \E_{\QQ}[d_{\textsf{d}}(U,V)]\big) \quad \QQ\in\Q_4^{\varepsilon} \Big\}.
\end{align}
We denote by $\mc{N}$ the set $\mc{N}^{\varepsilon}$ with $\varepsilon=0$.  
\end{definition}

\begin{theorem}[Nash equilibrium distortions]\label{theo:CheapTalkTheo}
\begin{align}
\forall \varepsilon\geq0, \;\forall n \in \N,  \quad &\textsf{NE}_{\varepsilon}^{\,n} \subset \mc{N}^{\varepsilon},\\
\lim_{\varepsilon \to 0}\lim_{n\to+\infty}&\textsf{NE}_{\varepsilon}^{\,n} = \mc{N}.
\end{align}
\end{theorem}

Theorem \ref{theo:CheapTalkTheo} is a consequence of Theorems  \ref{theo:EncoderCommitment} and \ref{theo:DecoderCommitment}. If the distribution $\PP_{U}  \QQ_{W|U}  \QQ_{V|W}$ have marginals that belong to the sets $\Q_{\textsf{d}}\big(\QQ_{UW}\big)$ and $\Q_{\textsf{e}}(\PP_{WV})$, then Shannon's encoding and decoding schemes form an $\varepsilon$-Nash equilibrium.

%
% three ingredients:
%1) $\lim_{\varepsilon \to 0}\mc{N}^{\varepsilon} = \mc{N}$, 
%2) the converse arguments of Theorems \ref{theo:EncoderCommitment}
% and \ref{theo:DecoderCommitment} are also valid for $\varepsilon$-best responses,
%3)  the random coding generates $\varepsilon$-best responses.

\begin{conjecture}\label{theo:CheapTalk}
\begin{align}
%\lim_{n\to+\infty}\lim_{\varepsilon \to 0}\;\textsf{NE}_{\varepsilon}^{\,n} = \mc{N}.\\
\lim_{n\to+\infty}\lim_{\varepsilon \to 0}\;\textsf{NE}_{\varepsilon}^{\,n} = \mc{N}.
\end{align}
\end{conjecture}

%%%%%%%%%%%%%%%%%%%%%%%%%%%%%%%%%%%%%%%%%%
%%%%%%%%%%%%%%%%%%%%%%%%%%%%%%%%%%%%%%%%%%
%\newpage
\appendices

%%%%%%%%%%%%%%%%%%%%%%%%%%%%%%%%%%%%%%%%%%
%%%%%%%%%%%%%%%%%%%%%%%%%%%%%%%%%%%%%%%%%%
%\section{Proof of Theorem \ref{theo:CheapTalkTheo}}\label{sec:proofTheoCheapTalk}
%
%XXX
%%%%%%%%%%%%%%%%%%%%%%%%%%%%%%%%%%%%%%%%%%
%%%%%%%%%%%%%%%%%%%%%%%%%%%%%%%%%%%%%%%%%%
\section{Preliminary results}\label{sec:ProofPreliminaryResults}

%%%%%%%%%%%%%%%%%%%%%%%%%%%%%%%%%%%%%%%%%%%%%%%%%%%%%%%%%%%%%%%%%%%%%%%%%%%%%%%%%%%%%%%%%%%%%%%%%%%%%%%%%%%%%%%%%%%%%%%%%%%%%%%%%%%%%%%%%%%%%%%%%%%%%%%%%%%%%%%%%%%%%%%%%%%%%%%%

\begin{definition}
Given $\PP_{UW}\in\Delta(\mc{U}\times \mc{W})$, tolerance $\delta>0$, let
\begin{align}
B_{\delta}(\PP_{UW})  =& \Big\{ \QQ_{UW}  \text{ s.t. }  ||\QQ_{UW} - \PP_{UW}||_1 \leq \delta \Big\}.
\end{align}
%We denote by $Q(u^n,w^n)\in\Delta(\mc{U}\times \mc{W})$ the empirical distribution of the sequences $(u^n,w^n)\in \mc{U}^n\times \mc{W}^n$ of length $n\in\N^{\star}$.
We define the set of typical sequences by
\begin{align}
%T_{\delta}(\PP_{U})  =& \Big\{ u^n \in \mc{U}^n \text{ s.t. }  Q_U^{n} \in B_{\delta}(\PP_{U})\Big\},\\
T_{\delta}(\PP_{UW})  =& \Big\{ (u^n,w^n) \text{ s.t. }  Q_{UW}^{n} \in B_{\delta}(\PP_{UW})   \Big\},
\end{align}
where $Q_{UW}^{n}$ denotes the empirical distribution of $(u^n,w^n)$.
%\begin{align}
%T_{\delta}(\PP_{UW})  =& \Big\{ (u^n,w^n) \text{ s.t. }  ||Q_{UW}^n - \PP_{UW}||_1 \leq \delta \Big\},\\
%T_{\delta}(\PP_{U})  =& \Big\{ u^n \in \mc{U}^n \text{ s.t. }  ||Q_{U}^n - \PP_{U}||_1 \leq \delta \Big\},
%\end{align}
\end{definition}

\begin{definition}
We consider two distributions $\PP_{U}\in\Delta(\mc{U})$, $\PP_{W}\in\Delta(\mc{W})$,  a rate parameter $\textsf{R}\geq0$ and a tolerance $\delta\geq0$. We define the sets
\begin{align}
\Q_{\delta}^-(\textsf{R}) &=  \Big\{ \QQ_{UW}\in\Delta(\mc{U}\times \mc{W}) \; \text{s.t.} \; ||\QQ_{U} - \PP_{U}||_1\leq\delta,\nonumber\\
  &||\QQ_{W} - \PP_{W}||_1\leq\delta  \;\; \text{and } \; \;   I(U;W) \leq \textsf{R}\Big\},\\
\Q_{\delta}^+(\textsf{R}) &=  \Big\{ \QQ_{UW}\in\Delta(\mc{U}\times \mc{W}) \; \text{s.t.} \; ||\QQ_{U} - \PP_{U}||_1\leq\delta,\nonumber\\
  &||\QQ_{W} - \PP_{W}||_1\leq\delta   \; \;\text{and } \; \;  I(U;W) \geq \textsf{R}\Big\}.
\end{align}
We use the notation $\Q_{0}^-(\textsf{R})$ and $\Q_{0}^+(\textsf{R})$ when $\delta=0$. 
\end{definition}

%\cite[Step 1.]{Lapidoth97}
\begin{lemma}[see Step 1 in \cite{Lapidoth97} and Lemma 4.3 in \cite{ScarlettGuillenSomekhMartinez_FTCIT20}]\label{lemma:DifferentPacking0}
We consider two distributions $\PP_{U}\in\Delta(\mc{U})$ and $\PP_{W}\in\Delta(\mc{W})$, a rate $\textsf{R}\geq0$, a small $\eta>0$ and $n\in\N^{\star}$. 
\begin{itemize}
\item[$\bullet$] We generate a sequence $U^n$ according to $\PP_{U}^{\otimes n}$.
\item[$\bullet$] Independently, we generate a family of sequences $\big(W^n(m)\big)_{m\in\{1,\ldots,2^{n\textsf{R}}\}}$ according to $\PP_{W}^{\otimes n}$.
\end{itemize}
There exists $\bar{\delta}$, for all $\delta<\bar{\delta}$ and for all $\varepsilon>0$, there exists $\bar{n}$, for all $n\geq\bar{n}$,
\begin{align*}
\prob\bigg( \exists  m  \in \{1,\ldots,2^{n\textsf{R}}\},\quad Q^n_m \in \Q_{\delta}^+(\textsf{R}+\eta) \bigg)  \leq \varepsilon,  %\nonu\label{eq:DifferentPacking} 
\end{align*}
where $Q^n_m$ denotes the empirical distribution of $(U^n,W^n(m))$.
\end{lemma}
The provide the proof of Lemma \ref{lemma:DifferentPacking0} in App. \ref{sec:LemmaProof}.

\begin{lemma}[Covering lemma, see Lemma 3.3 in \cite{ElGammalKim(book)11}]\label{lemma:covering0}
We consider a distribution $\PP_{UW}\in\Delta(\mc{U}\times \mc{W})$, a rate parameter $\textsf{R}= I(U;W)+\eta$ with $\eta>0$, $n\in\N$. 
\begin{itemize}
\item[$\bullet$] We generate a sequence $U^n$ according to $\PP_{U}^{\otimes n}$.
\item[$\bullet$] Independently, we generate a family of sequences $\big(W^n(m)\big)_{m\in\{1,\ldots,2^{n\textsf{R}}\}}$ according to $\PP_{W}^{\otimes n}$.
\end{itemize}
There exists $\bar{\delta}>0$, for all $\delta<\bar{\delta}$ and for all $\varepsilon>0$, there exists $\bar{n}$, such that for all $n\geq\bar{n}$,
\begin{align*}
\prob\bigg( \exists  m  \in \{1,\ldots,2^{n\textsf{R}}\},\quad ||Q^n_m -  \PP_{UW}||_1\leq\delta \bigg)  \geq 1 - \varepsilon.
\end{align*}
\end{lemma}

\begin{definition}
For $\PP_{U}\in\Delta(\mc{U})$, $\PP_{W}\in\Delta(\mc{W})$,  $\delta>0$, $\textsf{R}\geq0$, and $\textsf{D}\geq0$ we define
\begin{align}
&\Q_{\delta}(\textsf{R},\textsf{D}) =\Big\{ \QQ_{UW}\in\Delta(\mc{U}\times \mc{W}) \; \text{s.t.} \; ||\QQ_{U} - \PP_{U}||_1\leq\delta,\nonumber\\
&||\QQ_{W} - \PP_{W}||_1\leq\delta,\; 
  I(U;W) \leq \textsf{R},\;  \E \Big[d_{\textsf{e}}(U,V)\Big]\leq \textsf{D}\Big\}.
\end{align}
We have $\Q_{\delta}(\textsf{R},\textsf{D}) = \Q_{\delta}^-(\textsf{R})\cap \Q_{\delta}^{\circ}(\textsf{D}) $ with
\begin{align}
\Q_{\delta}^{\circ}(\textsf{D}) &=  \Big\{ \QQ_{UW}\in\Delta(\mc{U}\times \mc{W}) \; \text{s.t.} \; ||\QQ_{U} - \PP_{U}||_1\leq\delta,\nonumber\\
  &||\QQ_{W} - \PP_{W}||_1\leq\delta  \;\; \text{and } \; \;    \E \Big[d_{\textsf{e}}(U,V)\Big]\leq \textsf{D}\Big\}.
\end{align}
\end{definition}

%%%%%%%%%%%%%%%%%%%%%%%%%%%%%%%%%%%%%%%%%%
%%%%%%%%%%%%%%%%%%%%%%%%%%%%%%%%%%%%%%%%%%
\section{Achievability proof of Theorem \ref{theo:DecoderCommitment}}\label{sec:ProofTheoDecoder}

If the channel capacity is equal to zero, then a trivial coding scheme satisfies \eqref{eq:AchievabilityMD}. From now on, we assume that the channel capacity is strictly positive, therefore for all $\varepsilon_0>0$ there exists $\eta_0>0$ and a distribution $\PP_{WV}$ such that 
\begin{align}
\Big|D_{\textsf{d}}^{\star} -   \max_{\QQ_{UW}\in\Q^{\eta_0}_{\textsf{e}}(\PP_{WV})} \E_{\QQ_{UW} \atop  \PP_{V|W} } \Big[d_{\textsf{d}}(U,V)\Big]\Big|\leq\varepsilon_0,\label{eq:SolutionDecoder001}
\end{align}
where
\begin{align}
\Q^{\eta_0}_{\textsf{e}}(\PP_{WV}) =&\underset{\QQ_{UW}\in\Q_3^{\eta_0}(\PP_{W})}{\argmin}\;   \E_{\QQ_{UW} \atop  \PP_{V|W} } \Big[d_{\textsf{e}}(U,V)\Big],\\
\Q_3^{\eta_0}(\PP_{W}) =&  \Big\{ \QQ_{UW}\in\Delta(\mc{U}\times \mc{W}) \; \text{s.t.} \;\QQ_{U} = \PP_{U},\nonumber\\
 \QQ_{W} = \PP_{W}&\;\; \text{and} \;\; \max_{\PP_X} I( X; Y )  -   I( U ;W)   \geq 2\eta_0 \Big\}.
\end{align}

We use the notation $\QQ_{UW}$ to refer to the distribution that achieves the maximum in \eqref{eq:SolutionDecoder001}, and without loss of generality, we assume that $I( U ;W)= \max_{\PP_X} I( X; Y ) - 2\eta_0$. We introduce the rate parameter $\textsf{R}= I( U ;W)+\eta_0 $ and the tolerance of the typical sequences $\delta>0$. We consider that the decoder implements Shannon's channel decoding and lossy source decoding, see \cite[Sec. 3.1 and 3.6]{ElGammalKim(book)11}, that we denote by $\tau^{\star}$. We denote by $M$ and $m$ the  index selected by the encoder, whereas $\hat{M}$ and $\hat{m}$ refer to the  index selected by the decoder.
\begin{itemize}
\item[$\bullet$] The random codebooks $(W^n(m),X^n(m))_{m\in\{1,\ldots,2^{n\textsf{R} }\}}$ are drawn independently according to $\PP^{\otimes n}_W $ and $\PP^{\otimes n}_X $, where $\PP_X$ maximizes the channel capacity. 
\item[$\bullet$] The decoder observes the sequence of channel output $Y^n\in\mc{Y}^n$ and returns the unique index $\hat{m}$ such that the sequences  $\big(Y^n,X^n(\hat{m})\big) \in T_{\delta}(\PP_{X}\mc{T}_{Y|X})$ are jointly typical. Otherwise it returns the index $1$.
\item[$\bullet$] Then the decoder returns the sequence $W^n(\hat{m})$ corresponding to $\hat{m}$ and draws $V^n$ i.i.d. according to $\PP_{V|W}$.
\end{itemize}
Standard channel coding arguments ensures that 
\begin{align}
\exists \bar{\delta}_1,\forall \delta<\bar{\delta}_1,\forall \varepsilon_1>0,\exists \bar{n}_1\in \N^{\star},\forall n\geq \bar{n}_1, 
\;\prob(\hat{M}\neq M)\leq \varepsilon_1.\label{eq:ErrorPacking}
%\sum_{m\in\mc{M}}\prob(M=m)\prob(\hat{M}\neq m|M=m)\leq \varepsilon.
\end{align}
%\begin{remark}
%We remove the worst half of the codewords of $(X^n(m))_{m\in\{1,\ldots,2^{n\textsf{R} }\}}$ and we denote by $\mc{M}'$ the set of remaining indices. Then we have $|\mc{M}'|=2^{n(\textsf{R}-\frac{1}{n})}$ and
%\begin{align}
%\max_{m\in\mc{M}'}\prob(\hat{M}\neq m|M=m)\leq 2\varepsilon_1.
%\end{align}
%\end{remark}

Since the encoder is strategic, it selects a best response $\sigma \in \textsf{BR}_{\textsf{e}}(\tau^{\star})$ that, for a given $u^n$, returns $x^{n}$ in order to minimize
\begin{align}
& \sum_{y^n,v^n\atop \hat{m}}\mc{T}(y^n|x^n)\prob(\hat{m}|y^n)\PP^{\otimes n}(v^n|w^n(\hat{m}))\frac{1}{n} \sum_{t=1}^n d_{\textsf{e}}(u_t,v_t)\nonumber\\
 &=\sum_{\hat{m}}\prob(\hat{m}|x^n) \cdot \sum_{u,w} Q^n_{\hat{m}}(u,w)\sum_{v}
 \PP(v|w) d_{\textsf{e}}(u,v),\label{eq:OptimalXn}
 \end{align}
where $Q^n_{\hat{m}}\in\Delta(\mc{U}\times \mc{W})$ denotes the empirical distribution of $(u^n,w^n(\hat{m}))$. We denote by $x^{n\star}$ the sequence that minimizes \eqref{eq:OptimalXn} and we denote by 
\begin{align}
Q^{x^{n}} = \sum_{\hat{m}}\prob(\hat{m}|x^n) \cdot Q^n_{\hat{m}}\in\Delta(\mc{U}\times \mc{W}),\label{eq:AverageEmpiricalDistribution}
 \end{align}
the average empirical distribution induced by the input sequence $x^{n}$. By Lemma \ref{lemma:DifferentPacking0}, for all $\eta_2>0$, there exists $\bar{\delta}_2$, for all $\delta<\bar{\delta}_2$ and for all $\varepsilon_2>0$, there exists $\bar{n}_2$, for all $n\geq\bar{n}_2$,
\begin{align}
&\prob\bigg(Q^{X^{n\star}}\!\!\! \!  \notin \Q_{\delta}^-(\textsf{R}+\eta_2) \bigg) \leq \prob\bigg(Q^{X^{n\star}} \!\!\! \!\!\! \in \Q_{\delta}^+(\textsf{R}+\eta_2) \bigg)  \\
&+ \prob\bigg(||Q_U^{X^{n\star}} - \PP_U||_1+||Q_W^{X^{n\star}} - \PP_W||_1>\delta\bigg)  
\label{eq:DifferentPacking00} \\
\leq&\prob\bigg( \exists  x^n  \in \mc{X}^n,\quad Q^{x^{n}}  \in \Q_{\delta}^+(\textsf{R}+\eta_2) \bigg) + \varepsilon_2 \label{eq:DifferentPacking01} \\
\leq&\prob\bigg( \exists  m  \in \{1,\ldots,2^{n\textsf{R}}\},\; Q^n_m \in \Q_{\delta}^+(\textsf{R}+\eta_2) \bigg)  + \varepsilon_2\\
 \leq& 2\varepsilon_2.  %\nonu\label{eq:DifferentPacking} 
\end{align}

On the other hand, we assume that the encoder implements Shannon's coding scheme $\sigma_{\textsf{c}}$, by selecting the unique $m$ such that $(U^n,W^n(m))\in T_{\delta}(\QQ_{UW})$, and $1$ otherwise. By Lemma \ref{lemma:covering0}, there exists $\bar{\delta}_3>0$, for all $\delta<\bar{\delta}_3$ and for all $\varepsilon_3>0$, there exists $\bar{n}_3$, such that for all $n\geq\bar{n}_3$,
\begin{align}
\prob\bigg( \forall  m  \in \{1,\ldots,2^{n\textsf{R}}\},\quad ||Q^n_m -  \QQ_{UW}||_1>\delta \bigg)  \leq  \varepsilon_3.\label{eq:ErrorCovering}
\end{align}
The bounds given in \eqref{eq:ErrorPacking}, \eqref{eq:ErrorCovering} imply
\begin{align}
1 - \varepsilon_1- \varepsilon_3& \leq \prob\Big(Q^{X^{n}(m)}  \in  \Q_{\delta}^{\circ}(\textsf{D}+\mu) \Big)\\
&\leq \prob\Big(Q^{X^{n\star}}  \in  \Q_{\delta}^{\circ}(\textsf{D}+\mu) \Big),
\end{align}
with $\textsf{D}= \min_{\QQ_{UW}\in\Q_3^{\eta_0}(\PP_{W})} \E\big[d_{\textsf{e}}(U,V)\big]$ and $\mu= \delta  \overline{d_{\textsf{e}}}$ where $\overline{d_{\textsf{e}}}=\max_{u,v}d_{\textsf{e}}(u,v)$. Thus for all $\delta\leq\min(\bar{\delta}_1,\bar{\delta}_2,\bar{\delta}_3)$ and $n\geq\max(\bar{n}_1,\bar{n}_2,\bar{n}_3)$ we have
\begin{align}
&\prob\Big(Q^{X^{n\star}}\in \Q_{\delta}(\textsf{R}+\eta_2,\textsf{D}+\mu)\Big)\\
\geq&1 - \prob\Big(Q^{X^{n\star}}\notin \Q_{\delta}^-(\textsf{R}+\eta_2) \Big)-\prob\Big(Q^{X^{n\star}}\notin  \Q_{\delta}^{\circ}(\textsf{D}+\mu) \Big)\\
\geq&1-\varepsilon_1-2\varepsilon_2-\varepsilon_3.\label{eq:BoundProbaSet}
\end{align}
This shows the existence of a strategy $\tau^{\star}$ with codebook  $(w^n(m),x^n(m))_{m\in\{1,\ldots,2^{n\textsf{R} }\}}$ such that \eqref{eq:BoundProbaSet} is satisfied. We consider $\sigma \in \textsf{BR}_{\textsf{e}}(\tau^{\star})$ that achieves the maximum in \eqref{eq:MDProblem} and we denote $\overline{d_{\textsf{d}}}=\max_{u,v}d_{\textsf{d}}(u,v)$. Form Berge's Maximum Theorem the correspondance $(\delta,\textsf{R},\textsf{D}) \mapsto {\Q}_{\delta}(\textsf{R},\textsf{D})$ is continuous, and therefore
\begin{align}
&d_{\textsf{d}}^n(\sigma,\tau^{\star}) =    \E_{Q^{X^{n\star}}\atop  \PP_{V|W} } \Big[d_{\textsf{d}}(U,V)\Big] \\
\leq& \sup_{\PP_{UW}\in\atop\Q_{\delta}(\textsf{R}+\eta_2,\textsf{D}+\mu)}\E_{\PP_{UW}\atop  \PP_{V|W} } \Big[d_{\textsf{d}}(U,V)\Big] + (\varepsilon_1+2\varepsilon_2+\varepsilon_3)\overline{d_{\textsf{d}}}\\
\leq& \sup_{\PP_{UW}\in\atop\Q(\textsf{R}-\eta_0,\textsf{D})}\E_{\PP_{UW}\atop  \PP_{V|W} } \Big[d_{\textsf{d}}(U,V)\Big] + (\varepsilon_1+2\varepsilon_2+\varepsilon_3+\varepsilon_4)\overline{d_{\textsf{d}}}\\
=& \max_{\PP_{UW}\in\atop
\Q^{\eta_0}_{\textsf{e}}(\PP_{WV})} \E_{\PP_{UW} \atop  \PP_{V|W} } \Big[d_{\textsf{d}}(U,V)\Big] + (\varepsilon_1+2\varepsilon_2+\varepsilon_3+\varepsilon_4)\overline{d_{\textsf{d}}}\\
\leq&D_{\textsf{d}}^{\star} +\varepsilon_0 + (\varepsilon_1+2\varepsilon_2+\varepsilon_3+\varepsilon_4)\overline{d_{\textsf{d}}}.
\end{align}

We take $\varepsilon_0$, $\varepsilon_1$, $\varepsilon_2$, $\varepsilon_3$, $\varepsilon_4$, $\delta$, $\eta_2$, $\eta_0$ small and $n\in\N^{\star}$ large and the achievability result of Theorem \ref{theo:DecoderCommitment} follows. 

%%%%%%%%%%%%%%%%%%%%%%
\section{Proof of Lemma \ref{lemma:DifferentPacking0}}\label{sec:LemmaProof}

%\begin{proof}[Lemma \ref{lemma:DifferentPacking0}]
Lemma \ref{lemma:CoveringSetJointDistributions} below ensures for all $\delta>0$, there exists a family of distributions $(\QQ_{UW}^k)_{k\in \mc{K}} \subset \inte\Delta(\mc{U}\times \mc{W})$ with $|\mc{K}|<+\infty$ such that% \eqref{eq:LemmaCoverSimplex1} and \eqref{eq:LemmaCoverSimplex2}.
\begin{align}
\Delta(\mc{U}\times \mc{W}) &\subset \bigcup_{k\in\mc{K}} T_{\delta}(\QQ_{UW}^k),\label{eq:LemmaCoverSimplex1W}\\
\min_{k\in \mc{K}}\min_{(u,w)\in \mc{U}\times\mc{W}} \QQ^k(u,w)&\geq \frac{\delta}{4(|\mc{U}\times \mc{W}|-1)}.\label{eq:LemmaCoverSimplex2W}
\end{align}
Thus for all $\delta>0$, there exists a family of distributions $(\QQ_{UW}^{\tilde{k}})_{\tilde{k}\in \widetilde{\mc{K}}}  \subset\Q^+_{\delta}(\textsf{R}+\eta)   \cap\inte\Delta(\mc{U}\times \mc{W}) $ with $|\widetilde{\mc{K}}|<+\infty$ such that \eqref{eq:LemmaCoverSimplex2W} is satisfied and
\begin{align}
 \Q^+_{\delta}(\textsf{R}+\eta) \subset \bigcup_{\tilde{k}\in \widetilde{\mc{K}}} T_{\delta}(\QQ_{UW}^{\tilde{k}}).\label{eq:LemmaCoverSimplex3W}
\end{align}
We choose $\delta<\bar{\delta}$ such that $3 \bar{\delta} \log\frac{4(|\mc{U}\times \mc{W}|-1)}{\bar{\delta}}<\eta$. 
\begin{align}
&\prob\bigg( \exists  m  \in \{1,\ldots,2^{n\textsf{R}}\} \; \text{ s.t. }\; Q^n_m \in \Q^+_{\delta}(\textsf{R}+\eta) \bigg)   \label{eq:DifferentPackingProof1}\\ 
\leq &\prob\bigg( \exists  m   \; \text{ s.t. }\; Q^n_m \in  \bigcup_{\tilde{k}\in \widetilde{\mc{K}}} T_{\delta}(\QQ_{UW}^{\tilde{k}}) \bigg)   \label{eq:DifferentPackingProof1}\\ 
= &\prob\bigg(\exists  \tilde{k}\in \widetilde{\mc{K}},  \exists  m    \; \text{ s.t. }\; Q^n_m \in  T_{\delta}(\QQ_{UW}^{\tilde{k}}) \bigg)   \label{eq:DifferentPackingProof2}\\ 
\leq &\sum_{\tilde{k}\in \widetilde{\mc{K}} }\sum_{m\in\{1,\ldots,2^{n\textsf{R}}\}}\sum_{(u^n,w^n)\in \atop T_{\delta}(\QQ_{UW}^{\tilde{k}})}\PP_{U}^{\otimes n}(u^n)\PP_{W}^{\otimes n}(w^n) \label{eq:DifferentPackingProof3}\\ 
\leq &|\widetilde{\mc{K}} | \cdot  2^{n(\textsf{R} - I(U;W) + 3 \delta \log\frac{4(|\mc{U}\times \mc{W}|-1)}{\delta})}\label{eq:DifferentPackingProof4}\\ 
\leq &|\widetilde{\mc{K}} | \cdot 2^{-n(\eta- 3 \delta \log\frac{4(|\mc{U}\times \mc{W}|-1)}{\delta} )}\label{eq:DifferentPackingProof5}.
\end{align}
Equation \eqref{eq:DifferentPackingProof1} comes from \eqref{eq:LemmaCoverSimplex3W}. Equation \eqref{eq:DifferentPackingProof4} comes from \eqref{eq:LemmaCoverSimplex2W} with $ \min_{u,w} \QQ^{\tilde{k}}(u,w) \geq \frac{\delta}{4(|\mc{U}\times \mc{W}|-1)}$, and Proposition \ref{prop:TypicalSequences} and \ref{prop:SizeSetTypicalSequences} below. 
Equation \eqref{eq:DifferentPackingProof5} comes from $\QQ_{UW}^{\tilde{k}}\in  \Q^+_{\delta}(\textsf{R}+\eta) $, that induce $\textsf{R}\leq I(U;W)-\eta$.

Since $|\widetilde{\mc{K}} |<+\infty$ and $\eta-3 \delta \log\frac{4(|\mc{U}\times \mc{W}|-1)}{\delta}>0$, we choose $n$ large such that $|\widetilde{\mc{K}} | \cdot 2^{-n(\eta- 3 \delta \log\frac{4(|\mc{U}\times \mc{W}|-1)}{\delta} )}\leq \varepsilon$. This concludes the proof of Lemma \ref{lemma:DifferentPacking0}.
%\end{proof}

\begin{proposition}[see 1. pp. 27 in \cite{ElGammalKim(book)11}]\label{prop:TypicalSequences}
We consider $\PP_{U}\in\Delta(\mc{U})$, $n\in\N$, $\delta>0$. For all $u^n \in T_{\delta}(\PP_{U})$ we have
\begin{align}
2^{-n(H(U)+\delta_1)} \leq\PP_{U}^{\otimes n}(u^n) \leq 2^{-n(H(U)-\delta_1)},
\end{align}
with $\delta_1 =  \log\frac{1}{ \min\limits_{u\in \supp\PP_{U}} \PP(u)} \cdot  \delta$.
\end{proposition}

\begin{proposition}[see 2. pp. 27 in \cite{ElGammalKim(book)11}]\label{prop:SizeSetTypicalSequences}
We consider $\PP_{UW}\in\Delta(\mc{U}\times \mc{W})$, $n\in\N$, $\delta>0$. Then
%\begin{align}
$
\big|T_{\delta}(\PP_{UW})\big| \leq 2^{n(H(U,W)+\delta_2)}$
%\end{align}
with $\delta_2 =  \log\frac{1}{ \min\limits_{(u,w)\in \supp\PP_{UW}} \PP(u,w)} \cdot  \delta$.
\end{proposition}

\begin{lemma}\label{lemma:CoveringSetJointDistributions}
We consider a set $\mc{U}$ such that $2\leq|\mc{U}|<+\infty$. For all $\delta>0$, there exists a family of distributions $(\QQ_{U}^k)_{k\in \mc{K}}\subset \inte\Delta(\mc{U})$ with $|\mc{K}|<+\infty$ such that 
\begin{align*}
\Delta(\mc{U})\subset \bigcup_{k\in\mc{K}} T_{\delta}(\QQ_{U}^k),\quad %\label{eq:LemmaCoverSimplex1}
&\min_{k\in \mc{K}}\min_{u\in\mc{U}} \QQ^k(u)\geq \frac{\delta}{4(|\mc{U}|-1)}.\label{eq:LemmaCoverSimplex2}
\end{align*}
\end{lemma}
\begin{proof}[Lemma \ref{lemma:CoveringSetJointDistributions}]
We consider a symbols $\tilde{u}\in\mc{U}$ and we define the distributions 
\begin{align*}
\PP_{U} = \begin{cases}
1&\text{ if }U=\tilde{u},\\
0&\text{ otherwise, }
\end{cases}\;\;
\QQ_{U}^{\tilde{u}} = \begin{cases}
1-\frac{\delta}{4}&\text{ if }U=\tilde{u},\\
\frac{\delta}{4(|\mc{U}|-1)}&\text{ otherwise. }
\end{cases}
\end{align*}
Then,
\begin{align}
||\QQ_{U}^{\tilde{u}}- \PP_{U}||_1 =& \sum_{u}|\QQ^{\tilde{u}}(u) - \PP(u)| \nonumber\\
=& \frac{\delta}{4} + \frac{\delta}{4(|\mc{U}|-1)}(|\mc{U}|-1)  = \frac{\delta}{2}<\delta.
\end{align}
This shows that $\PP_{U}\in T_{\delta}(\QQ_{U}^{\tilde{u}})$. The same construction applies to any other symbol $\hat{u} \in\mc{U}$, and this generates a collection of distributions $(\QQ_{U}^{\hat{u}})_{\hat{u}\in\mc{U}}$. We construct a family of distributions $(\QQ_{U}^k)_{k\in \mc{K}}\subset \inte\Delta(\mc{U})$ based on the lattice with steps $\frac{\delta}{4(|\mc{U}|-1)}$ that connects the elements of $(\QQ_{U}^{\hat{u}})_{\hat{u}\in\mc{U}}$. Since $\Delta(\mc{U})\subset [0,1]^{|\mc{U}|-1}$, we have
%the family of distributions has a cardinality that satisfies 
$|\mc{K}| \leq \Big(\frac{4(|\mc{U}|-1)}{\delta}\Big)^{|\mc{U}|-1}<+\infty$.% and moreover  \eqref{eq:LemmaCoverSimplex2} are satisfied.
\end{proof}

\bibliographystyle{IEEEtran}
%\bibliography{/Users/maelletreust/Documents/Redaction/BiblioMael}

\begin{thebibliography}{10}

\bibitem{CrawfordSobel1982StrategicInformation}
V.~P. Crawford and J.~Sobel, ``Strategic information transmission,'' {\em
  Econometrica}, vol.~50, no.~6, pp.~1431--1451, 1982.

\bibitem{KamenicaGentzkow11}
E.~Kamenica and M.~Gentzkow, ``Bayesian persuasion,'' {\em American Economic
  Review}, vol.~101, pp.~2590--2615, 2011.

\bibitem{LeTreustTomala19}
M.~Le~Treust and T.~Tomala, ``Persuasion with limited communication capacity,''
  {\em Journal of Economic Theory}, vol.~184, p.~104940, 2019.

\bibitem{dughmi2019persuasion}
S.~Dughmi, R.~Niazadeh, A.~Psomas, and S.~M. Weinberg, ``Persuasion and
  incentives through the lens of duality,'' in {\em International Conference on
  Web and Internet Economics}, pp.~142--155, Springer, 2019.

\bibitem{SaritasYukselGezici2020}
S.~Sar{\i}ta\c{s}, S.~Y\"{u}ksel, and S. Gezici, ``Dynamic signaling games with quadratic criteria
  under nash and stackelberg equilibria,'' {\em Automatica}, vol.~115,
  p.~108883, 2020.

\bibitem{AkyolLangbortBasarIEEE17}
E.~Akyol, C.~Langbort, and T.~Ba\c{s}ar, ``Information-theoretic approach to
  strategic communication as a hierarchical game,'' {\em Proceedings of the
  IEEE}, vol.~105, no.~2, pp.~205--218, 2017.

\bibitem{LeTreustTomala(Allerton)16}
M.~Le~Treust and T.~Tomala, ``Information design for strategic coordination of
  autonomous devices with non-aligned utilities,'' {\em IEEE Proc. of the 54th
  Allerton conference, Monticello, Illinois}, pp.~233--242, 2016.

\bibitem{LeTreustTomala(IZS)18}
M.~Le~Treust and T.~Tomala, ``Strategic coordination with state information at
  the decoder,'' {\em Proc. of 2018 International Zurich Seminar on Information
  and Communication}, 2018.

\bibitem{LeTreustTomala_IT19}
M.~Le~Treust and T.~Tomala, ``Strategic communication with side information at
  the decoder,'' {\em [on-line] available: https://arxiv.org/abs/1911.04950},
  Nov. 2019.

\bibitem{Lapidoth97}
A.~Lapidoth, ``On the role of mismatch in rate distortion theory,'' {\em IEEE
  Transactions on Information Theory}, vol.~43, pp.~38--47, Jan. 1997.

\bibitem{ScarlettGuillenSomekhMartinez_FTCIT20}
J.~Scarlett, A.~G. i~F\`{a}bregas, A.~Somekh-Baruch, and A.~Martinez,
  ``Information-theoretic foundations of mismatched decoding,'' {\em
  Foundations and Trends in Communications and Information Theory}, vol.~17,
  no.~2--3, pp.~149--401, 2020.

\bibitem{VoraKulkarni_ISIT20}
A.~S. {Vora} and A.~A. {Kulkarni}, ``Achievable rates for strategic
  communication,'' in {\em 2020 IEEE International Symposium on Information
  Theory (ISIT)}, pp.~1379--1384, 2020.

\bibitem{VoraKulkarni_Arxiv2020}
A.~S. Vora and A.~A. Kulkarni, ``Information extraction from a strategic
  sender: The zero error case,'' {\em [on-line] available:
  https://arxiv.org/abs/2006.10641}, 2020.

\bibitem{JacksonSonnenschein07}
M.~O. Jackson and H.~F. Sonnenschein, ``Overcoming incentive constraints by
  linking decisions,'' {\em Econometrica}, vol.~75, pp.~241 -- 257, January
  2007.

\bibitem{ElGammalKim(book)11}
A.~E. Gamal and Y.-H. Kim, {\em Network Information Theory}.
\newblock Cambridge University Press, Dec. 2011.

\end{thebibliography}

\end{document}